\documentclass{article}

\usepackage{amsthm,amsfonts,braket,amsmath, fullpage}
\usepackage{xcolor}
\usepackage{tikz}
\usepackage[english]{babel}
\usepackage{tikz}
\usepackage{todonotes}
\usepackage{float}
\usetikzlibrary{shapes,shapes.geometric, automata,positioning, arrows}

\newcommand{\kb}[1]{\ket{#1}\bra{#1}}
\newcommand{\bk}[1]{\braket{#1|#1}}
\newcommand{\innerprod}[2]{\langle #1 | #2 \rangle}

\newcommand{\ds}{\displaystyle}

\newcommand{\e}{\wedge}

\newcommand{\Tr}{\texttt{Tr}}
\newcommand{\Real}{\texttt{Re}}

\newtheorem{protocol}{Protocol} 
\newtheorem{theorem}{Theorem}

\usepackage{subcaption}

\begin{document}

\title{Semi-Quantum Conference Key Agreement with GHZ-type states}

\author{R\'uben Barreiro$^{1}$\and Walter O. Krawec$^{2}$\and Paulo Mateus$^{1}$ \and Nikola Paunkovi\'c $^{1}$\and Andr\'e Souto$^{3}$}

\date{\footnotesize\textsuperscript{\textbf{1}} Instituto de Telecomunica\c{c}\~oes and Departamento de Matemática,\\ Instituto Superior T\'ecnico da Universidade de Lisboa \\
Av. Rovisco Pais 1, 1049-001, Lisboa, Portugal \\ \ \\ \textsuperscript{\textbf{2}}School of Computing, University of Connecticut\\
Storrs, CT USA\\ \ \\ \textsuperscript{\textbf{3}}Lasige and Faculdade de Ci\^encias da Universidade de Lisboa\\
Campo Grande 016, 1749-016, Lisboa, Portugal}

\maketitle

\begin{abstract}
We propose a semi-quantum conference key agreement (SQCKA) protocol that leverages on GHZ states. We provide a comprehensive security analysis for our protocol that does not rely on a trusted mediator party. We present information-theoretic security proof, addressing collective attacks within the asymptotic limit of infinitely many rounds. This assumption is practical, as participants can monitor and abort the protocol if deviations from expected noise patterns occur.  This advancement enhances the feasibility of SQCKA protocols for real-world applications, ensuring strong security without complex network topologies or third-party trust.
\end{abstract}

\section{Introduction}
\label{sec:introduction}

The primary threat quantum computing poses to classical public-key cryptography is still due to Shor's algorithm~\cite{sho:94}. This algorithm reduced to (quantum probabilistic) polynomial time the difficulty of breaking cryptographic systems commonly used nowadays, like RSA and elliptic-curve cryptography. Both problems are based on mathematical hardness assumptions that do not hold in quantum realm. Therefore, new quantum-resistant secure communication protocols and cryptographic applications have been extensively studied for the last decades. The first example of this was Wiesner's proposal  of unforgeable quantum money based on no-cloning feature of quantum mechanics~\cite{wiesner:1983} (proposed in early 70s, but only published in 1984).

Quantum mechanics brought to computer science, especially cryptography, a fresh perspective on dealing with the security and privacy of communications. In particular, replacing hardness mathematical assumptions, similar to the factoring problem mentioned above, by the laws of physics to ensure secrecy is the most relevant advantage of using a quantum-based approach to cryptography.

One of the most basic secure communication protocols aims to establish a common key between two parties, called \emph{key distribution}. Such a protocol enables the parties to communicate secretly using, for example, a one-time pad encryption scheme~\cite{vernam:sec-sig-sys:1919} -- a theoretically secure cryptographic scheme using a symmetric key. Diffie-Hellman protocol~\cite{dif:hel:76} is probably the most well-known and used example of a key distribution protocol. The underlying problem relies on the hardness of computing the discrete logarithm, but as mentioned previously, Shor's algorithm makes this protocol vulnerable to quantum computers as the discrete-log problem is also in bounded error quantum polynomial-time class (BQP). In 1984, Bennett and Brassard developed the first-ever quantum key distribution protocol (QKD), known as the BB84 protocol~\cite{ben:bra:84}, which was subsequently proven to be theoretically secure in~\cite{shor-preskill:sim-pro-sec-bb8-qua-key-dis-pro:2000}.  By now there are many different QKD protocols and the field has advanced rapidly both in theory and in practice~\cite{Massa2022}.  For a survey, the reader is referred to \cite{pirandola2020advances}.

Moving beyond the two-party scenario of QKD, so-called \emph{quantum conference key agreement} (QCKA) protocols allow for $n$ parties to establish a joint shared key (also called a \emph{conference key} or a \emph{group key}).  QCKA protocols have advanced rapidly from early work in the field \cite{cabello2000multiparty,epping2017multi} to newer protocols and security proofs \cite{grasselli2019conference,wu2016continuous,ottaviani2019modular} including experimental implementations \cite{proietti2021experimental}.  While two-party QKD protocols can be used to establish group keys by running multiple instances of the protocol between a ``leader'' and each other user, this often leads to inefficiencies \cite{epping2017multi}.  In particular, using a two-party QKD protocol to establish a group key would require that the number of instances of bipartite QKD protocols needed to establish a shared key between many parties does not scale optimally with the number of users. For example, using pairwise instances of QKD protocols between Alice and each Bob to distribute a pre-shared key between Alice and each Bob initially. Then Alice uses the pairwise key to wrap a final conference key generated on the fly to be common to all the parties. This solution scales linearly with the number of participants. Still, in the presence of eavesdropping and noise, each of the instances of QKD requires {\em pairwise} steps of information reconciliation and privacy amplification~\cite{bra:sal:94,ben:etal:95,ren:08}. Also, the number of quantum devices required is equal to the number of users. On the other hand, it also increases the burden of Alice constantly establishing several bipartite symmetric keys with Bobs to wrap the conference key and then distribute it to all of them using encryption algorithms like AES~\cite{dae:rij:00}, or (quantum) one-time pad~\cite{vernam:sec-sig-sys:1919,mos:tap:wol:00}.   For a recent survey on QCKA protocols, the reader is referred to \cite{murta2020quantum}.

Most of the proposed QCKA solutions solve the scalability problem but have the characteristic that the more parties are involved, the more quantum resources are needed, similar to QKD. To reduce the number of quantum devices necessary and to enable the possibility of seeing quantum cryptography as a practical service, the idea of a semi-quantum key distribution (SQKD)~\cite{boyer-kenigsber-mor:qua-key-dis-wit-cla-bob:2007,lu-cai:qua-key-dis-wit-cla-ali:2008} was established.  In such a scenario, one party is ``fully quantum'' in that he/she can manipulate quantum resources arbitrarily (e.g., can create any qubit state and measure in any basis).  The other party or parties are restricted to ``semi-quantum'' or almost ``classical'' operations.  Namely the semi-quantum users can only measure or send qubits in a single, fixed, basis (typically the computational $Z$ basis) or can only ``reflect'' qubits back to the sender, essentially disconnecting from the quantum channel, and sending any and all quantum states back to whoever sent them (without otherwise disturbing them or extracting information from them).

Numerous two-party SQKD protocols are known at this point (see \cite{iqbal2020semi} for a survey); there have also been interesting developments in semi-quantum versions of QCKA protocols.  Much work has been devoted to so-called ``fair key agreement'' protocols \cite{xu2023improvement,hong2024multiparty,yang2024efficient,zhou2020three,xu2022single} where the goal is to distill a conference key where the goal is to distill a conference key where all parties contribute to the randomness of the final raw key.  This is unlike many standard QKD protocols (including the one we discuss here), where raw keys are typically produced through a direct measurement of a state prepared by one party (or a third-party adversary).  Fair key protocols typically involve several additional steps on top of the measurement to ensure randomness is added by all parties.  The protocol we consider in this paper is more like standard (S)QKD in that the raw key will be produced from a measurement of a GHZ state produced by one of the parties.  Ensuring randomness of the final secret key will be guaranteed, then, through standard privacy amplification.

Beyond this notion of fair semi-quantum key agreement, there have also been several ``standard''  semi-quantum QCKA protocols proposed, without the requirement that all parties contribute randomness to the final key similar to ours.  Some works involve protocols operating on cycle topologies whereby qubits travel from party to party in a cycle network~\cite{xian2009quantum,ye2023circular}.  Such protocols allow for a group key to be established with single qubits; however, the security analysis is difficult since Eve has multiple attempts to interact with the channel, and in general security of these protocols is only proven against certain attacks (\cite{ye2023circular} included an information-theoretic security proof, but only for the two-party version of their protocol).  Furthermore, reflection events may pick up a large amount of noise due to the multi-hop setting (a qubit must travel through all parties before returning to the original sender).  Other semi-quantum QCKA protocols have used cluster states of the form $\ket{0000} + \ket{0110} + \ket{1001} - \ket{1111}$ (up to normalization) for three-party key agreement~\cite{zhou2019multi} (an $n$-party extension was proposed also in that reference, however, it required a cycle network again).  An interesting use of GHZ states was proposed in~\cite{pan2024multi}, which allowed for the establishment of multiple keys between all pairs of parties.  GHZ states were also used in~\cite{zhu2018semi} for a two-party SQKD protocol.

\section{Preliminaries}
\label{sec:preliminaries}
By $A$, $B$, $B_i$, and $E$ we denote the relevant parties: Alice, Bobs, the i-th Bob, and Eve, respectively. Let $\rho_A$ be a quantum state over system $A$'s Hilbert space $\mathcal{H}_A$. Its von Neumann entropy is defined as $S(A)_{\rho_A} = -\text{Tr}(\rho_A\log_2\rho_A)$. Given a  bipartite state $\rho_{AE}$ over the corresponding Hilbert space $\mathcal{H}_{AE}$, its conditional von Neumann entropy is defined as $S(A|E)_{\rho_{AE}} = S(AE)_{\rho_{AE}} - S(E)_{\rho_{E}}$. It is well known that if both systems are classical (i.e., $\rho_{AE} = \sum p_{i,i'} \ket{i}\bra{i}_A\otimes \ket{i'}\bra{i'}_E$, where $\langle{i}|{i'}\rangle = \delta{ij}$), $S(A|E)_{\rho_{AE}} = H(A|E)$, where $H(A|E) = H(A,E) - H(E)$ is the conditional Shannon entropy of system $A$ given system $E$. 

Like QKD, QCKA protocols use quantum and authenticated classical channels to establish a \emph{raw key}. To establish a raw key of size $M$, the QCKA process requires $N \ge M$ rounds (in general, $N = \tau p_{\texttt{rk}}M$, where $\tau$ is the transmission coefficient, a probability that a qubit reaches its destination, and $p_{\texttt{rk}}$ is the probability that a round when a qubit reaches its destination is used to establish a bit of the raw key).  These raw keys are classical bit strings held by all the participants in the protocol, and due to noise and eavesdropping, they are partially correlated and only partly secret. To eliminate the errors, as mentioned previously one has to use an error correction protocol~\cite{bra:sal:94}; to turn it more private, one has to use a privacy amplification protocol on top of the raw key~\cite{ben:etal:95,ren:08}. These processes establish a final secret common key of $\ell$ bits. In this paper, we focus on bounding the \emph{effective key rate} $r = \ell/N$ in the asymptotic case of infinitely many rounds, i.e., when $N \rightarrow \infty$, using an adaptation of the Devetak-Winter key-rate equation~\cite{QKD-Winter-Keyrate,QKD-renner-keyrate}:
\begin{equation}
  r \geq S(A|E)_\rho - H(A|B).
\end{equation}
Note that the state $\rho$ used to establish the key on every round may depend on Alice, Bob's, and Eve's. In our case, as we are working with group-key protocols, we will use the following asymptotic version of the key rate from \cite{grasselli2018finite},
\begin{equation}\label{eq:key-rate}
    r \geq S(A|E)_\rho - \max_j H(A|B_j),
\end{equation}
where the maximum is taken over all of the Bobs. Namely, the key rate of a group-key protocol is simply the difference between Eve's uncertainty on Alice's register and the \emph{maximal} uncertainty between Alice's key bits and Bobs' (see \cite{grasselli2018finite}).  The second term is due to error correction leakage, and Alice must send an error-correcting code to fix the errors in the ``noisiest'' Bob (i.e., the Bob $B_j$ with the most uncertainty on Alice's key bit); see ~\cite{mur:etal:20} for more details. This value can be directly estimated using Alice and Bob's measurement results. The first term, quantum entropy, is usually the most challenging quantity to estimate in a (S)QKD security proof through the positivity of the key rate.  To bound the first term of Equation~\eqref{eq:key-rate} $S(A|E)_\rho$, we use the following result, proved in~\cite{QKD-Tom-Krawec-Arbitrary}, to estimate the von Neumann entropy: 

\begin{theorem}\label{thm:cq-entropy}
  (From \cite{QKD-Tom-Krawec-Arbitrary}): Let $\rho_{AE}$ be a quantum state of the form:
  \begin{equation}
    \rho_{AE} = \frac{1}{\mathcal{N}}\sum_{a = 0}^1\kb{a}_A\otimes\left(\sum_{i=1}^{K}\kb{E_i^a}_E\right),
  \end{equation}
where $\mathcal{N}$ is a normalization constant, $\ket{E_i^a}$ are in general neither normalized nor mutually orthogonal vectors. Then, the von Neumann entropy $S(A|E)_\rho$ may be bounded by
  \begin{eqnarray}\label{eq:S_bound}
    S(A|E)_\rho &\ge \ds\frac{1}{\mathcal{N}}\sum_{i=1}^{K}\left(\bk{E_i^0} + \bk{E_i^1}\right)
    \times\left[H\left(\frac{\bk{E_i^0}}{\bk{E_i^0} + \bk{E_i^1}}\right) - H(\lambda_i)\right],
  \end{eqnarray}
  \text{where}
  \begin{eqnarray}\label{eq:lambda}
    \lambda_i = \frac{1}{2}\left(1 + \frac{\sqrt{\big(\bk{E_i^0} - \bk{E_i^1}\big)^2 + 4\textnormal{\Real}^2\braket{E_i^0|E_i^1}}}{\bk{E_i^0} + \bk{E_i^1}}\right).
  \end{eqnarray}
where $H(\lambda)$ is the Shannon entropy of the binary distribution $\{\lambda, 1-\lambda\}$ and $E_i^0$ and $E_i^0$ are eavesdroppers subspaces for Alice's outputs 0 and 1, respectively.
\end{theorem}
Note that, in the above, the order of the summation of the terms in $\rho_{AE}$ does not matter - all orderings will give a lower-bound.  Thus, technically, one should maximize the right-hand-side of Equation \ref{eq:S_bound} over all possible orderings, or permutations, of Eve's states.  We will do this later when evaluating the key-rate expression.

In the remainder of the paper, we use GHZ-type states~\cite{ghz} of the form
\begin{equation}
\label{eq:ghz-like}
    \ket{g^{(n+1)}(x,y)} = \frac{1}{\sqrt{2}}\ket{0}\ket{y} + \frac{(-1)^x}{\sqrt{2}}\ket{1}\ket{\bar{y}},
\end{equation} 
where $x\in\{0,1\}$ and $y\in\{0,1\}^{n}$.  Notice that the set of these states is called the GHZ basis over $(n+1)$ qubits. Here, $\bar{y}$ is the bit-wise complement of $y$.

\section{The proposed semi-quantum conference key agreement protocol}
\label{sec:proposal}

Our proposal combines SQKD and QCKA to develop a multiparty key distribution for semi-quantum contexts. We assume that there are $n$ ``classical'' Bobs and one fully quantum Alice. Our protocol consists of $N$ rounds (for simplicity, we  assume that qubits do not get absorbed during transmission, i.e., in our ideal case $\tau = 1$, and we have that the size of teh raw key is $M = p_{\texttt{rk}} N$).  Each round of the protocol starts with Alice who prepares $(n+1)$ qubits in a GHZ-like state of the form~\eqref{eq:ghz-like}. She then keeps one qubit with her (called the zeroth qubit), while sending one qubit to each Bob $B_i$, with $i \in \{ 1,2, \dots n \}$.  Like all the other semi-quantum proposals mentioned above, all parties pre-share a secret key. This key allows them to share a random $N$-bit string  ${\bf\Theta} = ({\bf\Theta}_1{\bf\Theta}_2 \dots {\bf\Theta}_N)$. In each round $j \in \{1,2, \dots N\}$, if ${\bf\Theta}_j = 1$, all Bobs \texttt{Measure and Resend} their qubits, i.e., each Bob measures his corresponding qubit in the computational (``classical'') $Z$ basis and resends it back to Alice, who also measures her qubit in the $Z$ basis, thus establishing a bit of the raw key between the $(n+1)$ parties (in order to check the effects of noise and eavesdropping, she additionally measures all qubits received from Bobs in the $Z$ basis and compare the results with her own); otherwise, all Bobs \texttt{Reflect} their corresponding qubits back to Alice who checks whether the $(n+1)$ qubits are still in the original GHZ-like state (she performs a quantum measurement by projecting onto the GHZ-like state). Note that in practical implementations, the qubits are encoded in photon states, and photons get absorbed when measured. Thus, to resend a measured photon each Bob prepares a new one in the state corresponding to the measurement outcome. A single round $j \in \{1,2, \dots N\}$ of the protocol is depicted in Figure~\ref{fig:schematic_semi_quantum_conference_key_agreeement}.

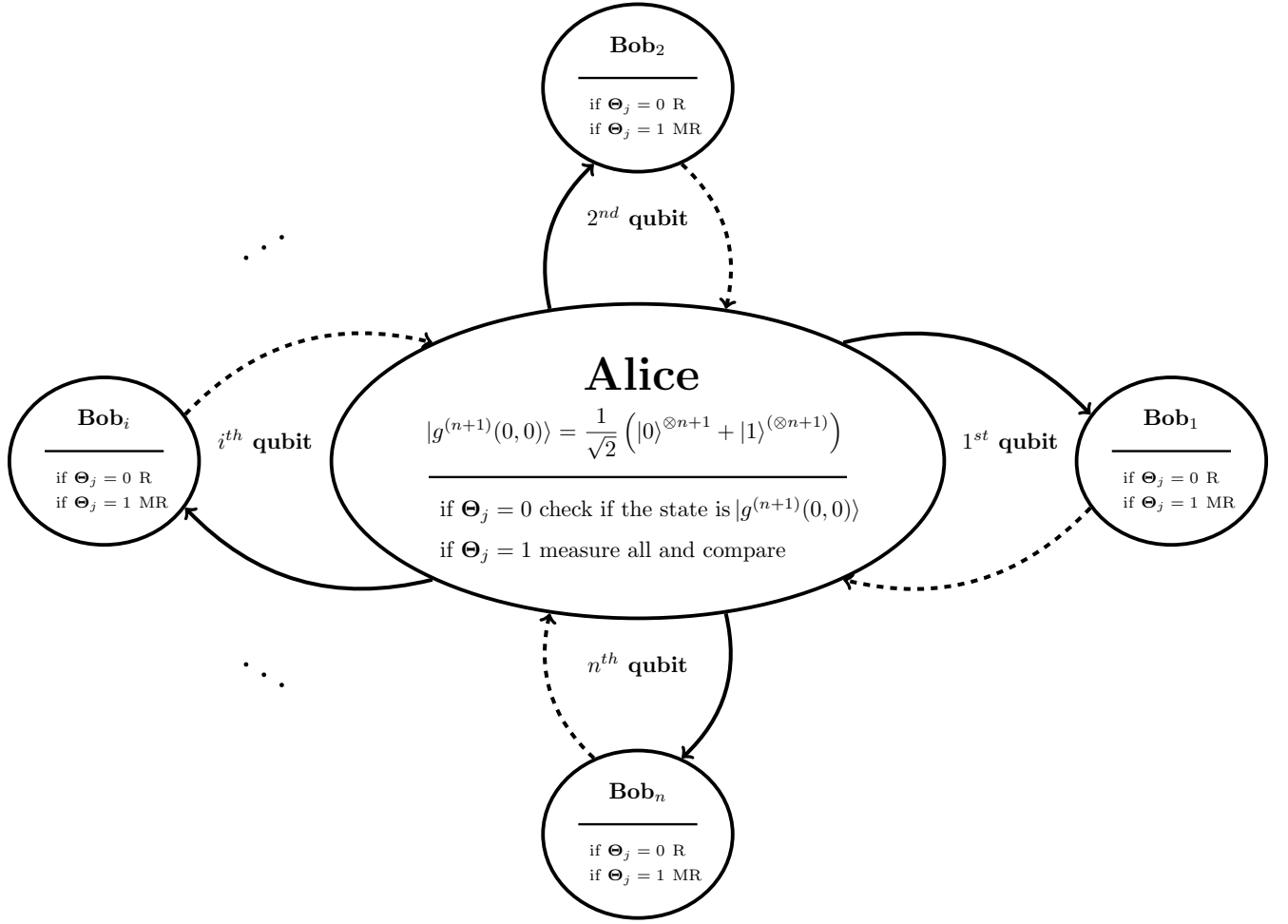
\begin{figure}[!th]
\begin{center}
\scalebox{.9}{
\begin{tikzpicture}[node distance=2cm] 
   \node[ellipse, draw, ultra thick, text width = 6.5cm, align=center] (alice)  
   {
       {
           {\bf\huge Alice}\\[2.5mm]
           \hspace{-1mm}$ \ket{g^{(n+1)}(0,0)}= \ds\frac{1}{\sqrt 2} \left(\ket{0}^{\otimes n+1} + \ket{1}^{(\otimes n+1)}\right)$\\[2mm]
           \rule{\textwidth}{1pt}\\[2.5mm]  
           {$\begin{array}{ll} \textrm{if} \!\!\!\!
        & {\bf\Theta}_j = 0 \textrm{ check if the state is}  \ket{g^{(n+1)}(0,0)}\\[2mm]
        \textrm{if} \!\!\!\!
        & {\bf\Theta}_j = 1 \textrm{ measure all and compare}
        \end{array}$}
       }
    };

 \node[draw=none,rectangle,below left=2.5cm of alice, rotate=-30] {{\huge$\dots$}};
 
  \node[draw=none,rectangle,above left=2.5cm of alice, rotate=30] {{\huge$\dots$}};
 
   \node[ellipse, draw, ultra thick,text width = 1.85cm, align=center] (bob1)  [right=of alice] 
   {$\textbf{Bob}_1$\\
   
   \rule{\textwidth}{1pt}\\[2.5mm]  
       {
        \scriptsize $\begin{array}{ll}
        \textrm{if} & \!\!\!\!\! {\bf\Theta}_j = 0 \ \textrm{R}\\[1mm]
        \textrm{if} & \!\!\!\!\! {\bf\Theta}_j = 1 \ \textrm{MR}
        \end{array}$
       }
    };  
   \node[ellipse, draw, ultra thick,text width = 1.85cm, align=center ] (bob2)  [above=2cm of alice] 
   { $\textbf{Bob}_2$\\
   
  \rule{\textwidth}{1pt}\\[2.5mm]  
       {
        \scriptsize $\begin{array}{ll}
        \textrm{if} & \!\!\!\!\! {\bf\Theta}_j = 0 \ \textrm{R}\\[1mm]
        \textrm{if} & \!\!\!\!\! {\bf\Theta}_j = 1 \ \textrm{MR}
        \end{array}$
       }
    };  
     
   \node[ellipse, draw, ultra thick,text width = 1.85cm, align=center ] (bob3)  [left= of alice] 
   {{\bf $\textbf{Bob}_i$}\\
   
   \rule{\textwidth}{1pt}\\[2.5mm]  
       {
        \scriptsize $\begin{array}{ll}
        \textrm{if} & \!\!\!\!\! {\bf\Theta}_j = 0 \ \textrm{R}\\[1mm]
        \textrm{if} & \!\!\!\!\! {\bf\Theta}_j = 1 \ \textrm{MR}
        \end{array}$
       }
    };

   \node[ellipse, draw, ultra thick,text width = 1.85cm, align=center ] (bob4)  [below=2cm of alice] 
   {$\textbf{Bob}_n$\\
   
   \rule{\textwidth}{1pt}\\[2.5mm]  
   
       {
        \scriptsize $\begin{array}{ll}
        \textrm{if} & \!\!\!\!\! {\bf\Theta}_j = 0 \ \textrm{R}\\[1mm]
        \textrm{if} & \!\!\!\!\! {\bf\Theta}_j = 1 \ \textrm{MR}
        \end{array}$
       }
    }; 

\path[->] 
(alice)   edge[bend left, anchor=center, above,  ->,line width=1.7pt] node[text width = 2cm, anchor = center, above] {} (bob1)
 
(alice)   edge[bend left, anchor=center, above, ->,line width=1.7pt] node[sloped, anchor = center, above] {} (bob2)
  
(alice)   edge[bend left, anchor=center, above, ->,line width=1.7pt] node[sloped, anchor = center, above] {} (bob3)
    
 (alice)   edge[bend left, anchor=center, above, ->,line width=1.7pt] node[sloped, anchor = center, above] {} (bob4)
 
(alice)   edge[swap, draw = none, above] node[sloped, anchor = center, above] {{\bf$1^{st}$ qubit} } (bob1)

(alice)   edge[swap, draw = none, above] node[sloped, anchor = center, rotate = -90, above] {{\bf $2^{nd}$ qubit} } (bob2)

(alice)   edge[swap, draw = none, above] node[sloped, anchor = center, above] {{\bf $i^{th}$ qubit }} (bob3)

(alice)   edge[swap, draw = none, above] node[sloped, anchor = center, rotate = 90, above] {{\bf $n^{th}$ qubit} } (bob4)
 
(bob1)   edge[->,line width=1.7pt, bend left,  anchor=center, above, dashed] node[sloped, text width = 4.1cm, anchor = center, below] {} (alice)

(bob2)   edge[->,line width=1.7pt, bend left,  anchor=center, above, dashed] node[sloped, text width = 4.1cm, anchor = center, below] {} (alice)
     
(bob3)   edge[->,line width=1.7pt, bend left,  anchor=center, above, dashed] node[sloped, text width = 4.1cm, anchor = center, below] {} (alice)
     
(bob4)   edge[->,line width=1.7pt, bend left,  anchor=center, above, dashed]  node[sloped, text width = 4.1cm, anchor = center, below] {} (alice)
     ;
\end{tikzpicture}
}
\end{center}
\caption{Schematic drawing of our proposal for a single round. In the beginning, Alice prepares a GHZ state with $n$+1 particles and sends one particle of that state to each of $n$ Bob's (full arrows). 
Using a pre-shared key, each Bob decides to Reflect (R) the state or Measure and Resend (MR) back the outcome as a state (dashed arrows). When Alice receives all the states from Bob, she either confirms that she has GHZ state in the first case or measures her qubit and compares it to all the other outcomes of measuring the remaining ones in the computational basis.}
\label{fig:schematic_semi_quantum_conference_key_agreeement}
\end{figure}

Below, we present in detail our SQCKA protocol.

\begin{protocol}[\texttt{SQCKA}]\label{pro:1}
\ \\

    \textbf{\textit{Setup}:}
    \begin{enumerate}
        \item Number of parties $(n+1)$, one fully quantum Alice and $n$ classical Bobs.
        \item $N$, number of rounds.
        \item A pre-shared key (of size $<N$) that all the parties extend to $N$ bits.
    \end{enumerate}
    
    \textbf{\textit{Steps}:}
    
    \begin{enumerate}
       \item Alice prepares $\ket{g^{(n+1)}(0,\vec 0)}^{\otimes N}$, with $\vec 0 \equiv (00\dots 0)$ being a string of $n$ zeroes (and analogously for~$\vec 1)$. For each round $j \in \{1,2, \dots N\}$ she keeps the zeroth qubit with herself, and sends each qubit $i \in \{1,2, \dots n\}$ to the corresponding Bob $B_i$.
        
        \item Alice and Bobs share a random $N$-bit string  ${\bf\Theta} = ({\bf\Theta}_1{\bf\Theta}_2 \dots {\bf\Theta}_N)$ obtained using their secret pre-shared random string.

        \item In each round $j \in \{1,2, \dots N\}$, if ${\bf\Theta}_j = 1$, each Bob measures the received qubit in the $Z$ basis and resends it to Alice (called a \texttt{SIFT} round); otherwise, all Bobs reflect their qubits back to Alice (called a \texttt{CTRL} round). 

    	\item In each round $j \in \{1,2, \dots N\}$, upon receiving back $n$ qubits from Bobs, Alice performs her measurement. If ${\bf\Theta}_j = 1$, she measures her zeroth qubit, as well as each of Bobs' $n$ qubits, in the $Z$ basis, and compares the results (in the ideal noiseless case all parties should obtain the same $(n+1)$ results -- Alice's single-qubit results, and one result per each Bob from their measurements depicted in the above point 3 -- and therefore share a bit of the raw key); otherwise, she projects all the $(n+1)$ qubits onto the GHZ-like state $\ket{g^{(n+1)}(0, \vec 0)}$ (and should verify the projection with certainty in the ideal noiseless~case).
\end{enumerate}
\end{protocol}

To optimize the resources, it is desirable to decrease the size of the pre-shared key, such that it is shorter than the length $N$ of the string ${\bf\Theta}$. There are standard techniques to achieve this goal. A possible way to do this is by biasing relative frequencies of bits from string ${\bf\Theta}$, such that ${\bf\Theta}_j = 1$ occurs more often than ${\bf\Theta}_j = 0$ -- since ${\bf\Theta}_j=0$ rounds are used to check eavesdropping and perform parameter estimation, by fixing security level $\varepsilon$, one can define the number of parameter estimation rounds $\ell (\varepsilon, N)$, such that when  $N \rightarrow \infty$ we have $\ell (\varepsilon,N)/N \rightarrow 0$, that is, $\ell (\varepsilon, N)\in o(N)$. By encoding the indices $j$ to be used for the \texttt{CTRL} rounds only (for which ${\bf\Theta}_j = 0$), one thus needs a pre-shared key considerably shorter than the total number of rounds $N$ (note that by knowing the positions of the \texttt{CTRL} rounds one automatically knows the positions of the \texttt{SIFT} rounds, thus knowing the whole string ${\bf\Theta}$). One can further improve this by allowing a variable security level that would increase with $N$. For example, one can choose to have $\sqrt N$ \texttt{CTRL} rounds, costing $\sqrt N{\log}_{2}(N) << N'< N$ bits of the pre-shared key, where $N'$ is the size of the final secure key obtained after classical post-processing, namely privacy amplification and information reconciliation.

Upon performing the above step of the protocol, that involve quantum operations, they proceed with a fully classical final stage of the protocol. This consists of two parts. The first is parameter estimation, for which the results from the $\texttt{CTRL}$ rounds are used, as well as some from the $\texttt{SIFT}$ rounds. The latter are chosen by cut-and-choose technique, in the same manner as in the standard QKD protocols. The second part of the classical information post-processing is privacy amplification and information reconciliation (error correction) performed on the results obtained in the remaining $\texttt{SIFT}$ rounds, resulting in the final secure key. Since this classical post-processing stage is studied in the literature independently from quantum cryptography, as it is equally applicable to classical communication schemes that well predate quantum, we omit its detailed analysis, as it is out of the scope of the current study.

It is straightforward to show that our protocol is sound, i.e., in the case of no eavesdropping and no noise, the parties will share a common key at the end of the protocol.

\begin{theorem}
Protocol \ref{pro:1} is sound.
\end{theorem}

\begin{proof}
Since there is no interference from the environment, we know that the state of the quantum registers at the beginning of the protocol is:
\begin{equation}
    \ket{g^{(n+1)}(0, \vec 0)} = \frac{1}{\sqrt 2} (\ket{0\cdots 0} + \ket{1\cdots 1}).
\end{equation}
Alice sends a particle $i$  of this entangled state to the $i$-th Bob.

Let us first consider the \texttt{CTRL} cases, i.e., the cases where ${\bf\Theta}_j = 0$, and hence that at the $j$-th iteration of the protocol, all Bobs reflect the received qubits to Alice. Since there is no noise, by all Bobs reflecting, it means that when Alice collects all the qubits, the entire system's state is again the $\ket{g^{(n+1)}(0,\vec 0)}$ state. Alice can confirm that this is the case by performing, for example, a fidelity test.

In the second case, i.e., in the case of a \texttt{SIFT} round, each Bob will measure the qubit in the computational basis. Since the state is a GHZ state and the measurement basis is the same for all Bobs, the outcomes are all equal. Therefore, the state that Alice should receive back from all Bobs is either all 0's or all 1's. By measuring in the computational basis all the qubits received and her own qubit, she can verify what bit should be added to the raw key. 
\end{proof}

In the following sections, we discuss, in detail, the security proof of our protocol in a realistic scenario where Eve, the eavesdropper, may introduce noise into the system.

\section{Security proof}
\label{sec:security}
In our security analysis, the eavesdropper is restricted to collective attacks, i.e., the attacks are made on each iteration of the protocol independently and identically, but it is possible to perform a joint measurement of private ancillary. Note that this assumption can be removed using a {\em de Finetti-type} argument~\cite{QKD-renner-keyrate,chr:kon:ren:09,ren:07}.

To simplify the argument of the security proof, we consider a single round (out of $N$) of the protocol in which Alice sends $n$ transferring qubits, one to each Bob. Each Bob $B_i$, with $i \in \{ 1, 2, \dots n \}$, performs one of the two operations on his qubit, depending on the value of the bit $\Theta_{B_i}$, and sends his transferring qubits back to Alice. After receiving Bobs qubits, Alice performs the $GHZ/Z$-measurements, according to the value of her bit $\Theta_A$. Eve can entangle with the transferring qubits twice, once in the forward direction and again in the backward direction, using her register $E$. Bobs and Alice record their measurement results in their respective registers $B$ and $M_A$.

For simplicity and reader's reference, we will be using the following notation:
\begin{itemize}
	\item {\em Indices}: Throughout the manuscript, the ranges of all indices are explicitly denoted when used. Typically, indices $i \in \{ 1, 2, \dots n \}$ and $j \in \{ 1, 2, \dots N \}$ label Bobs ($n$ of them) and rounds ($N$ of them), respectively. Index $a$ denotes a bit, $a \in \{ 0,1 \}$, while indices $b,b',c$ and $c'$ denote $n$-bit strings, $b,b',c,c' \in \{ 0, 1, \dots 2^n - 1 \}$. Exceptionally, in particular in Section~\ref{sec:depolarising_channel}, indices $i,j$ and $k$  denote $n$-bit strings as well, $i,j,k \in \{ 0, 1, \dots 2^n - 1 \}$. Finally, the $n$-bit strings labeling the two prominent GHZ states are denoted as $\vec a \equiv (aa\dots a)$, with $a \in \{ 0,1 \}$.
    \item $A$: Alice's qubit (in her lab all the time).
    \item $T$: The transferring $n$ qubits (from Alice to Bobs and back to Alice). The initial state of the $AT$ register is $\ket{g^{(n+1)}(0,\vec 0)}_{AT}$.
    \item $B$: Bobs' memory -- $n$ qubits, initially in the state $\ket{\vec{0}}_B = \ket{0}_{B_1} \cdots \ket{0}_{B_n}$. Analogously, we define states $\ket{b}_{T} \equiv \ket{b_1}_{T_1}\ket{b_2}_{T_2}\dots \ket{b_n}_{T_n}$ for arbitrary bit string $b = (b_1b_2\dots b_n)$. 
    \item $\Theta = (\Theta_A\Theta_B)$: Alice's bit $\Theta_A$ and Bobs' $n$ bits $\Theta_B = (\Theta_{B_1}\Theta_{B_2}\dots \Theta_{B_n})$ determining the Agents' action ($\Theta_A = 0$ for Reflection and $\Theta_A = 1 $ for Measure and Resend, and analogously for $\Theta_{B_i}$) in a {\em single} run of the protocol, say run $j \in \{1,2, \dots N \}$ (for simplicity, instead of $\Theta^{(B_i)}_j$, we write $\Theta_{B_i}$, etc.). Thus, we have $\Theta_A = \Theta_{B_1} = \Theta_{B_2} = \dots = \Theta_{B_n} = {\bf\Theta}_j$. Different labels $A$ and $B_i$ are introduced to denote different physical systems in which the bit value ${\bf\Theta}_j$ is stored for different agents (typically, one uses classical systems to store bit values $\Theta_A$ and $\Theta_{B_i}$; here, we assume a more general case of quantum systems, either microscopic qubits or macroscopic stable memories). eavesdropping. 
    
        In addition, by $\Theta = (\Theta_A\Theta_B)$ we denote Hilbert spaces of Alice's and Bobs' corresponding qubits. Note that for simplicity we slightly abuse the notation, as here $\Theta_A$ and $\Theta_B$ are not bit(-string)s but Hilbert space labels. Thus, the initial $(n+1)$-qubit state determining the agents' action is:
    \begin{eqnarray}
        \ket{\varphi_0}_\Theta &=& \ket{g^{(n+1)}(0,\vec 0)}_\Theta \nonumber\\ &=& \frac{1}{\sqrt 2} \left(\ket 0_{\Theta_A}\ket 0_{\Theta_{B_1}}\cdots \ket 0_{\Theta_{B_n}} + \ket 1_{\Theta_A}\ket 1_{\Theta_{B_1}}\cdots \ket 1_{\Theta_{B_n}}\right) \equiv \frac{1}{\sqrt 2} \left(\ket{0\vec{0}}_{\Theta} + \ket{1\vec{1}}_{\Theta}\right).
    \end{eqnarray}

    \item $E$: Eve's register in default initial state $\ket {\Omega}_{E}$.

    \item $M_A$: Alice's memory. $M_A = M_A^{(g)}M_A^{(01\cdots n)}$; where $M_A^{(g)}$ represents a qubit that stores the result of the $GHZ$-measurement, while $M_A^{(01\cdots n)}$ represent $(n+1)$ qubits that store the results of the $Z$ measurements.
\end{itemize}

\subsection{State evolution}
\label{sec:state_evolution}

The first step of the evolution of the initial state is given by (for the time being, we omit Alice's memory~$M_A$)
\begin{eqnarray}
    \ket{\psi_0}_{A T E B \Theta} &=& \frac{1}{\sqrt 2} \left(\ket 0_{A}\ket {\vec 0}_{T} +  \ket 1_{A}\ket {\vec1}_{T}\right)\ket{\Omega}_E\ket{\vec 0}_{B} \ket{\varphi_0}_\Theta \nonumber\\
    &\mapsto& I_A \otimes U^{(1)}_{TE} \otimes I_{B} \otimes I_{\Theta} \ket{\psi_0}_{A T E B \Theta} \equiv \ket{\psi_1}_{A T E B \Theta},
    \label{eq:forward_action}
\end{eqnarray}
where Eve's entangling action $U^{(1)}_{TE}$, defined here only on relevant vectors, is
\begin{eqnarray}
\label{eq:eve_forward_unitary}
    U^{(1)}_{TE} \left( \ket{\vec 0}_T \ket{\Omega}_{E} \right) &=& \sum_{b=0}^{2^n-1} \sqrt{p(b | 0)} \ket{ b}_T \ket{E_{0b}}_{E}\nonumber\\
    U^{(1)}_{TE} \left( \ket{\vec 1}_T \ket{\Omega}_{E} \right) &=& \sum_{b=0}^{2^n-1} \sqrt{p(b | 1)} \ket{ b}_T \ket{E_{1b}}_{E}.
\end{eqnarray}
Note that by assumption vectors $\ket{ b}_T$ represent an orthonormal basis, but while $\ket{E_{0b}}_{E}$ and $\ket{E_{1b}}_{E}$ are normalized to $1$, they are, in general, {\em not orthogonal} to each other. 

Thus, we have the following state upon transferring qubits reached Bobs, and their action (Reflect, or Measure and Resend), described as
\begin{eqnarray}
\label{eq:bob_action}
    \ket{\psi_1}_{A T E B \Theta} &=&  \frac{1}{\sqrt 2} \left(\ket 0_{A}\sum_{b=0}^{2^n-1} \sqrt{p(b | 0)} \ket{ b}_T \ket{E_{0b}}_{E}+  \ket 1_{A}\sum_{b=0}^{2^n-1} \sqrt{p(b | 1)} \ket{b}_T \ket{E_{1b}}_{E}\right)\ket{\vec 0}_{B} \ket{\varphi_0}_\Theta \nonumber\\
    &\mapsto& I_A \otimes I_{\Theta_A} \otimes V_{TB\Theta_B} \otimes I_{E} \ket{\psi_1}_{A T E B \Theta} \equiv \ket{\psi_2}_{A T E B \Theta}.
\end{eqnarray}
As above, Bobs' action $V_{TB\Theta_B}$ is defined only on relevant subspaces
\begin{eqnarray}
\label{eq:bob_unitary}
    V_{TB\Theta_B} &=& \ket{\vec 0}_{\Theta_B} \bra{\vec 0} \otimes I_{TB} + \ket{\vec 1}_{\Theta_B} \bra{\vec 1} \otimes \left[\sum_{b=0}^{2^n-1} \ket{b}_T \bra{b} \otimes \big(\ket b_B\bra{\vec 0}  + h.c. \big)\right].
\end{eqnarray}
Thus, the overall state upon Bobs' operation is (note the above mentioned notation ambiguity, as $\Theta$ here denotes both a bit value, as well as the agents' corresponding Hilbert spaces; also, to simplify the expressions, we assumed that both Bobs' actions are equally probable, even though, as noted above, the $\Theta = 1$ case of raw key generation is in practice overwhelmingly more probable)
\begin{eqnarray}
\label{eq:bob_action_state}
    \ket{\psi_2}_{A T E B \Theta} &=&  \frac{1}{\sqrt 2} 
    \left(
    \ket{0\vec 0}_{\Theta}\frac{1}{\sqrt 2} \left[ \ket 0_{A}\sum_{b=0}^{2^n-1} \sqrt{p(b | 0)} \ket{ b}_T \ket 0_{B}\ket{E_{0b}}_{E} + \ket 1_{A}\sum_{b=0}^{2^n-1} \sqrt{p(b | 1)} \ket{ b}_T \ket 0_{B}\ket{E_{1b}}_{E} \right]\right. \nonumber \\
    & &\quad\quad + \left.   
    \ket{1\vec 1}_{\Theta}\frac{1}{\sqrt 2} \left[ \ket 0_{A}\sum_{b=0}^{2^n-1} \sqrt{p(b | 0)} \ket{ b}_T \ket b_{B}\ket{E_{0b}}_{E} + \ket 1_{A}\sum_{b=0}^{2^n-1} \sqrt{p(b | 1)} \ket{ b}_T \ket b_{B}\ket{E_{1b}}_{E} \right]
    \right)\nonumber\\
    &=& 
    \frac{1}{\sqrt 2} \sum_{\Theta = 0}^1 \ket{\Theta\vec{\Theta}}_\Theta\frac{1}{\sqrt 2} \sum_{a = 0}^1 \ket{a}_A \sum_{b=0}^{2^n -1} \sqrt{p(b | a)} \ket{b}_T \ket {b\e \vec{\Theta}}_{B} \ket{E_{ab}}_{E}.
\end{eqnarray}
Upon sending back the transferring qubits to Alice, Eve entangles again, resulting in the following state
\begin{eqnarray}
    \ket{\psi_3}_{A T E B \Theta} &=&  I_\Theta \otimes I_A \otimes I_{B} \otimes I_{E_1} \otimes U^{(2)}_{TE}\ket{\psi_2}_{A T E B \Theta},
\label{eq:backward_action}
\end{eqnarray}
in which
\begin{equation}
\label{eq:eve_bakward_unitary}
    U^{(2)}_{TE}(\ket{b}_{T}\ket {E_{ab}}_{E}) = \sum_{b'=0}^{2^n-1} \sqrt{p'(b' | ab)} \ket{b'}_T \ket {E_{abb'}}_{E}.
\end{equation}
Hence, the overall state upon transferring qubits are back in Alice's lab is
\begin{equation}
\label{eq:alice_measurement_state}
    \ket{\psi_3}_{A T E B \Theta} =  \frac{1}{\sqrt 2} \sum_{\Theta = 0}^1 \ket{\Theta\vec{\Theta}}_\Theta\frac{1}{\sqrt 2} \sum_{a = 0}^1 \ket{a}_A \sum_{b,b'=0}^{2^n -1} \sqrt{p(b | a) p'(b'|ab)} \ket{b'}_T \ket {b\e \vec{\Theta}}_{B} \ket{E_{abb'}}_{E}.
\end{equation}
We note that according to our model, given by Equations~\eqref{eq:forward_action},~\eqref{eq:bob_action} and~\eqref{eq:backward_action}, Alice's qubit $A$ is isolated in her lab. Of course, one could straightforwardly redo our calculation to incorporate the noise acting on Alice's qubit $A$, but it is to be expected that in practical realizations, the combined action of Eve and the environment on the transferring qubits would, in general, be stronger than that of the environment only on Alice's qubit $A$. Therefore, for simplicity, we neglect the effects of noise on Alice's qubit $A$.

Finally, Alice performs her measurement on state~\eqref{eq:alice_measurement_state}, either in the $Z$ basis, or projecting on the $n+1$-qubit GHZ state. Adding Alice's memory, initially in the state $\ket{00\vec 0}_{M_A}$ (note that register $M_A$ consists of $(n+2)$ qbits), we have
\begin{equation}
\label{eq:state3}
   \ket{\psi_3}_{A T E B \Theta M_A} = \frac{1}{2} \sum_{\Theta = 0}^1 \sum_{a = 0}^1 \sum_{b,b'=0}^{2^n -1} \sqrt{p(b | a) p'(b'|ab)}\ket{\Theta\vec{\Theta}}_\Theta \ket{00\vec 0}_{M_A}  \ket{a}_A \ket{b'}_T \ket {b\e \vec{\Theta}}_{B} \ket{E_{abb'}}_{E}.
\end{equation}
Alice's final measurement is given by 
\begin{eqnarray}
     U_{\Theta_A (A T)_{M_A} } =  \ket{0}_{\Theta_A} \bra{0}\otimes V^{(GHZ)}_{(A T)_{M_A}} + \ket{1}_{\Theta_A} \bra{1}\otimes V^{(Z)}_{(A T)_{M_A}},
\end{eqnarray}
where
\begin{eqnarray}
    V^{(GHZ)}_{(A T)_{M_A}} &=& \left[\ket{GHZ^{n+1}}_{AT}\bra{GHZ^{n+1}} \otimes I_{M^{(g)}_A} + \right.\nonumber\\
    && \left.\quad\quad (I_{AT} - \ket{GHZ^{n+1}}_{AT}\bra{GHZ^{n+1}} )\otimes (\ket 1_{M^{(g)}_A}\bra 0 + \ket 0_{M^{(g)}_A}\bra 1) \right]\otimes I_{M_A^{(01\cdots n)}}
\end{eqnarray}
and
\begin{equation}
    V^{(Z)}_{(A T)_{M_A}} = I_{M^{(g)}_A} \otimes  \left(\sum_{a = 0}^1 \sum_{b = 0}^{2^n -1} \ket{ab}_{AT}\bra{ab}\otimes (\ket{ab}_{M_A^{(01 \cdots n)}}\bra{0\vec 0} + h.c.)\right),
    \label{eq:alice_z_measurement}
\end{equation}
where $h.c.$ stands for the hermitian conjugate. 

\subsection{Parameter estimation}
\label{sec:parameter_estimation}

Note that up to now, all kets were normalized to $1$ and orthogonal to each other, with the exception of Eve's kets $\ket{E_{ab}}_{E_1}$ and $\ket{E_{b'}}_{E_2}$ which, in general, may not be orthogonal to each other. Let us define
\begin{eqnarray}
    \ket{\widetilde E_{abb'}}_{E} \equiv \sqrt{p'(b' | ab) }\ket{E_{abb'}}_{E}.
\end{eqnarray}

Thus, the state in Equation~\eqref{eq:state3} can be rewritten as
\begin{eqnarray}\label{eq:state-before-measure}
    \ket{\psi_3}_{A T E B \Theta} &=&  
    \frac{1}{2} \sum_{\Theta = 0}^1 \sum_{a = 0}^1 \sum_{b,c = 0}^{2^n -1} \sqrt{p(b | a)} \ket{\Theta}_{\Theta} \ket{a}_A \ket{c}_T \ket {b\e \Theta}_{B} \ket{\widetilde E_{abc}}_{E}\otimes \ket{0}_{M_A}.
\end{eqnarray}

{\em First case:} $\Theta = 0$, which means that all Bobs are reflecting and Alice is measuring $\ket{GHZ^{n+1}}_{AT}$. Upon Bobs' reflection and the arrival of the transferring qubits to Alice's lab, but before her final projection measurement, the state is (note that for simplicity we omitted $\ket{0}_{\Theta}$ to which the $\Theta$ register is projected, as well as $\ket {b\e 0}_{B} = \ket {0}_{B}$ and $\ket{0}_{M_A}$, which factor out)
\begin{eqnarray}
    \ket{\psi_4}_{A T E} &=&  
    \frac{1}{\sqrt 2} \left[ \ket{0}_A \sum_{b,c = 0}^{2^n -1}  \sqrt{p(b | 0)} \ket{c}_{T} \ket{\widetilde E_{0  bc}}_{E} + 
  \ket{1}_A \sum_{b,c = 0}^{2^n -1}  \sqrt{p(b | 1)} \ket{c}_{T} \ket{\widetilde E_{1 bc}}_{E}\right].
\end{eqnarray}
Hence, by denoting 
\begin{eqnarray}
\label{eq:defofE0c}
    \sum_{b = 0}^{2^n -1}  \sqrt{p(b | 0)} \ket{\widetilde E_{0  bc}}_{E}\equiv \ket{\widetilde{E}_{0c}}_E
\end{eqnarray}
and
\begin{eqnarray}
\label{eq:defofE1c}
    \sum_{b = 0}^{2^n -1} \sqrt{p(b | 1)} \ket{\widetilde E_{1  bc}}_{E}\equiv \ket{\widetilde{E}_{1c}}_E,
\end{eqnarray}
we can rewrite the last state as:
\begin{eqnarray}
    \ket{\psi_4}_{A T E} &=&  
    \frac{1}{\sqrt 2} \left[ \ket{0}_A \sum_{c = 0}^{2^n-1}\ket{c}_{T} \ket{\widetilde E_{0c}}_{E} + \ket{1}_A \sum_{c = 0}^{2^n-1}\ket{c}_{T} \ket{\widetilde E_{1c}}_{E}\right]. 
\end{eqnarray}

\noindent Now, 
\begin{eqnarray}
\label{eq:p_GHZ_npas}
    p_{GHZ} &=& ||_{AT}\innerprod{GHZ^{n+1}}{\psi_4}_{ATE}||^2 \\
            &=&  {}_{AT}\innerprod{GHZ^{n+1}}{\psi_4}_{ATE}\innerprod{\psi_4}{GHZ^{n+1}}_{AT}
\end{eqnarray}
and noticing that  $\ket{GHZ^{n+1}}_{AT} = \frac{1}{\sqrt 2} (\ket 0_A \ket{\vec 0}_T + \ket 1_A \ket{\vec 1}_T)$ we have
\begin{eqnarray}
{}_{AT}\innerprod{GHZ^{n+1}}{\psi_4}_{ATE} = \frac{1}{2} \left(\ket{\widetilde E_{0  0}}_{E} + \ket{\widetilde E_{1  1}}_{E}\right)
\end{eqnarray}
and hence
\begin{eqnarray}
\label{eq:p_GHZ}
    p_{GHZ} &=& \frac{1}{4}\left[ \innerprod{\widetilde E_{0 0}}{\widetilde E_{0 0}}+ \innerprod{ \widetilde E_{1 1}}{\widetilde E_{1 1}} +  2\Real(\innerprod{\widetilde E_{0 0}}{ \widetilde E_{1 1}})\right].
\end{eqnarray}

Note that $p_{GHZ}$ is an observable quantity, obtained by Alice's projection onto the $\ket{GHZ^{n+1}}_{AT}$ state. Moreover, by performing cut-and-choose for the case of $\Theta = 0$ and performing the measurement in the $Z$ basis on the transferring qubits $T$, Alice also obtains $\innerprod{\widetilde E_{0 c}}{\widetilde E_{0 c}}$ and $\innerprod{\widetilde E_{1 c}}{\widetilde E_{1 c}}$, in particular $\innerprod{\widetilde E_{0 0}}{\widetilde E_{0 0}}$ and $\innerprod{\widetilde E_{1 1}}{\widetilde E_{1 1}}$. Thus, from~\eqref{eq:p_GHZ} one estimates $\Real (\innerprod{\widetilde E_{0 0}}{ \widetilde E_{1 1}})$, which should feature in the expression for the bounds for von Neumann conditional entropy.

Therefore, from~\eqref{eq:p_GHZ} we have
\begin{equation}
\label{eq:real_part_1}
   \Real (\innerprod{\widetilde E_{0 0}}{ \widetilde E_{1 1}}) = \frac{4 p_{GHZ} - \innerprod{\widetilde E_{0 0}}{\widetilde E_{0 0}}+ \innerprod{ \widetilde E_{1 1}}{\widetilde E_{1 1}}}{2}.
\end{equation}

Further, let $|\alpha_{0 0}|^2 = \innerprod{\widetilde E_{0 0}}{\widetilde E_{0 0}}$ and $|\alpha_{1 1}|^2 = \innerprod{\widetilde E_{1 1}}{\widetilde E_{1 1}}$. Since by definition $||\ket{E_{00}}|| = ||\ket{E_{11}}|| = 1$, from Equations~\eqref{eq:defofE0c} and~\eqref{eq:defofE1c}, we have:
\begin{eqnarray}
    |\alpha_{a a}|^2 & = & \innerprod{\widetilde E_{a a}}{\widetilde E_{a a}} \nonumber \\
    &=& \sum_{b,b' = 0}^{2^n -1} \sqrt{p(b | a) p(b' | a)} \innerprod{ \widetilde E_{a b a}}{\widetilde E_{a b' a}} \\
    &=& \sum_{b,b' = 0}^{2^n -1} \sqrt{p(b | a) p(b' | a) p'(a | a b) p'(a| a b')} \innerprod{ E_{a b a}}{E_{a b' a}}.\nonumber 
\end{eqnarray}

\noindent In general,
\begin{eqnarray}
\label{eq:alphas}
    |\alpha_{a c}|^2 & = & \innerprod{\widetilde{E}_{a c}}{\widetilde{E}_{a c}} \nonumber \\
    &=& \sum_{b,b' = 0}^{2^n -1} \sqrt{p(b | a) p(b' | a)} \innerprod{ \widetilde E_{a b c}}{\widetilde E_{a b' c}} \\
    &=& \sum_{b,b' = 0}^{2^n -1} \sqrt{p(b | a) p(b' | a) p'(c | a b) p'(c| a b')} \innerprod{ E_{a b c}}{E_{a b' c}}.\nonumber
\end{eqnarray}
Since from the additional cut-and-choose step for the $\Theta = 1$ case (see below) we have estimates for $p(b|a)$, and $|\alpha_{a c}|^2 = \innerprod{\widetilde{E}_{a c}}{\widetilde{E}_{a c}}$ are observable quantities. By solving~\eqref{eq:alphas} one obtains $\innerprod{ \widetilde E_{a b c}}{\widetilde E_{a b' c}}$ that feature in the expressions for the bounds on conditional entropy $S(A|E)$, see~\eqref{eq:new-entropy-bound} and~\eqref{eq:new-lambda}. These equations are derived by adapting Theorem~\ref{thm:cq-entropy} to our protocol featuring $n$ Bobs.

Also, we have 
\begin{equation}
\label{eq:real_parts}
    \Real\innerprod{\widetilde E_{0 0}}{\widetilde E_{1 1}} = \sum_{b,b' = 0}^{2^n -1} \sqrt{p(b|0) p(b'|1)} \innerprod{\widetilde E_{0 b 0}}{\widetilde E_{1 b' 1}},
\end{equation}
which gives an alternative estimate of $\Real\innerprod{\widetilde E_{0 0}}{\widetilde E_{1 1}}$, in addition to the one obtained by~\eqref{eq:real_part_1}.

\bigskip

{\em Second case:} $\Theta = 1$, which means that Bobs are measuring their memory qubits $B$ in the $Z$ basis and resending, while Alice is, upon the transferring qubits $T$ arrive in her lab, measure her qubit $A$ in the $Z$ basis. She discards the $T$ qubits.  Note that a more optimal solution would be to measure the $A$ qubit and all $T$ qubits, using the ``lowest noise'' measurement for her raw key (thus minimizing the uncertainty between Alice and the Bobs).  However, here we take the simpler approach and just use the $A$ register - our analysis below can be easily adapted to using an alternative $T$ qubit.

Upon Bobs receive the transferring qubits, and correlating them with their memory $B$, the state is (see Equation~\eqref{eq:bob_action})
\begin{eqnarray}
    \ket{\psi_4}_{A T B E} & = & \frac{1}{\sqrt 2} \sum_{a = 0}^1 \ket{a}_A \sum_{b=0}^{2^n -1} \sqrt{p(b | a)} \ket{ b}_T \ket {b}_{B} \ket{E_{ab}}_{E} \nonumber \\
    & = & \frac{1}{\sqrt 2} \sum_{b=0}^{2^n -1} \left( \sqrt{p(b | 0)} \ket{0}_A \ket{E_{0b}}_{E}
 + \sqrt{p(b | 1)} \ket{1}_A \ket{E_{1b}}_{E} \right)\ket{bb}_{TB}.
\end{eqnarray}
Thus, the probability of obtaining bit-string $b$ when measuring the memory $B$ (equivalently, the transferring qubits $T$) is given by
\begin{eqnarray}
    p_B(b) = \left\| \frac{1}{\sqrt 2} \left( \sqrt{p(b | 0)} \ket{0}_A \ket{E_{0b}}_{E}
 + \sqrt{p(b | 1)} \ket{1}_A \ket{E_{1b}}_{E} \right)\right\|^2 . 
\end{eqnarray}
Noting that $\innerprod{a}{a'}_A = \delta_{a,a'}$, we have
\begin{eqnarray}
    p_B(b) = \frac{1}{2} \left( p(b | 0)
 + p(b | 1) \right). 
\end{eqnarray}

Upon the arrival of the transferring qubits, sent by Bobs, to Alice's lab, but before her final projection measurement, the state is\begin{eqnarray}
    \ket{\psi_4}_{A T B E} &=& 
    \frac{1}{\sqrt 2} \sum_{a = 0}^1 \sum_{b,c = 0}^{2^n -1} \sqrt{p(b | a)p'(c|ab)}  \ket{a}_A \ket{c}_T \ket {b}_{B} \ket{E_{abc}}_{E}.
\end{eqnarray}
Noting that $\innerprod{a}{a'}_A = \delta_{a,a'}$, $\innerprod{c}{c'}_T = \delta_{c,c'}$ and $\innerprod{E_{abc}}{E_{abc}}_E = 1$, one obtains Alice's observable probabilities of the measurements in the computational basis, $p_{A}(a)$, for the case of measuring only her qubit $A$, and $p_{A}(ac)$, for the case of measuring both her qubit $A$ and the transferring qubits $T$: 
\begin{eqnarray}
\label{eq:cond_prob_1}
    p_{A}(a) &=& \frac{1}{ 2} \sum_{b,c = 0}^{2^n -1} p'(c | ab) p(b | a) = \sum_c \ [P'(a) P]_{ca} = \frac{1}{ 2}, \nonumber \\ 
    p_{A}(ac) &=& \frac{1}{ 2} \sum_{b = 0}^{2^n -1} p'(c | ab) p(b | a) = [P'(a) P]_{ca},
\end{eqnarray}
where $[P]_{ba} = p(b|a)/\sqrt{2}$ and $[P'(a)]_{cb} = p'(c|ab)/\sqrt{2}
$.

Note that Alice's qubit was always in her possession, and all the operations performed were on other systems. Thus, Alice's probability of the $Z$ measurement does not depend on the moment when the measurement is performed and is always $p_A(0) = p_A(1) = 1/2$.

Also, by adding additional cut-and-choose step in which Alice communicate her outcome to Bobs for a certain subset of runs (for the case when $\Theta = 1$), Bobs can obtain the conditional probabilities $p(b|0)$ and~$p(b|1)$.

\subsection{Entropy evaluation}
\label{sec:entropy_evaluation}

We now compute $S(A|E)$, the entropy in the state resulting in a key-bit being derived.  From Equation~\eqref{eq:state3}, we derive the final mixed state assuming all parties measure their systems in the $Z$ basis (also, conditioning on $\Theta = 1$, which is required for the parties to have a key-bit on a particular round.  This final mixed state is found to be:
\begin{equation}
  \rho_{ATBE} = \frac{1}{2}\sum_{a=0}^1\sum_{b,b' = 0}^{2^n -1} p(b|a)p'(b'|ab)\kb{a b' b}_{ATB}\otimes\kb{E_{abb'}}_E,
\end{equation}
where, at this point, Alice holds both the $A$ and the $T$ register.  Since Alice's key bit is derived from her first qubit register (the $A$ register), we trace out the $T$ and $B$ registers, leaving us with: \begin{align}
  \rho_{AE} &= \frac{1}{2}\kb{0}_A\otimes\sum_{b,b' = 0}^{2^n -1} p(b|0)p'(b'|0b)\kb{E_{0bb'}}_E + \frac{1}{2}\kb{1}_A\otimes\sum_{b,b' = 0}^{2^n -1} p(b|1)p'(b'|1b)\kb{E_{1bb'}}_E\notag.
  \end{align}

From the above, we may use Theorem \ref{thm:cq-entropy} to compute the conditional von Neumann entropy as follows:

\begin{align}
  S(A|E)_\rho &= \frac{1}{2}\max_{\pi_1,\pi_2}\sum_{b,b' = 0}^{2^n -1}\Bigg[\left(p(b|0)p'(b'|0b) + p(\pi_1(b)|1)p'(\pi_2(b')|1\pi_1(b))\right)\notag\\
  &\times\left(H\left(\frac{p(b|0)p'(b'|0b)}{p(b|0)p'(b'|0b) + p(\pi_1(b)|1)p'(\pi_2(b')|1\pi_1(b))}\right) - H\Big(\lambda\big(b \,b'\,\pi_1(b)\,\pi_2(b')\big)\Big)\right)\Bigg],\label{eq:new-entropy-bound}
\end{align}
where the above maximization is over all permutations $\pi_1,\pi_2:\{0, 1, \cdots, 2^n-1\}\rightarrow\{0,1,\cdots, 2^n-1\}$ and:
\begin{equation}
  \lambda(b\, b'\, c\, c') = \frac{1+\sqrt{(p(b|0)p'(b'|0b) - p(c|1)p'(c'|1c))^2 + 4p(b|0)p'(b'|0b)p(c|1)p'(c'|1c)\Real^2\!\braket{E_{0bb'}|E_{1cc'}}}}{2\left[p(b|0)p'(b'|0b) + p(c|1)p'(c'|1c)\right]}.\label{eq:new-lambda}
\end{equation}

We will now consider a detailed analysis for a relevant noise case, namely the depolarization channel. We note that our above analysis is general and can be applied to any collective attack. However, deriving a suitable bound on the inner-products of the various Eve vectors is, in general, non-trivial.  We were able to derive suitable bounds in the depolarization case and do so below.  Note that depolarizing noise is the most commonly considered noise model in theoretical QKD security proofs; furthermore, the parties can even enforce that the noise is depolarizing by aborting if they detect something different (thus, the protocol remains secure even if the noise is not depolarizing since parties will simply abort in that case).  Directly analyzing alternative noise models remains an interesting future problem.

\section{Noise evaluation -- depolarising channel}
\label{sec:depolarising_channel}

In the following, we evaluate a lower bound for the secret key rate in the asymptotic scenario assuming the noise and eavesdropping effect to be modeled by the depolarising channel acting on an arbitrary state $\rho_T$~as 
\begin{equation}
\label{eq:depolarising}
	\mathcal{E}_{Q}(\rho_T) \equiv \Tr_E\left[U_{TE}^{(Q)} \big(\rho_T \otimes \ket{\Omega}_{E}\bra{\Omega}\big) U_{TE}^{(Q)\dagger}\right] = (1-Q)\rho_T + \frac{Q}{d}I_T. 
\end{equation} 
Note that we are interested in the action of the forward channel $\mathcal{E}_{Q}$ only on the initial GHZ-type state components $\ket{\vec a}$, where $\vec a \equiv (a,a, \dots a)$ and $a \in \{ 0,1 \}$, which give the relevant probabilities from Equation~\eqref{eq:eve_forward_unitary}. Thus,~\eqref{eq:depolarising} reduces to
\begin{equation}
\label{eq:depolarising_F_1}
\mathcal{E}_{Q}(\ket{\vec a}_T\bra{\vec a}) \equiv \Tr_E\left[U_{TE}^{(Q)} \ket{\vec a \, \Omega}_{TE}\bra{\vec a \, \Omega} U_{TE}^{(Q)\dagger}\right] = (1-Q)\ket{\vec a}_{T}\bra{\vec a} + \frac{Q}{d}I_T , 
\end{equation}
with
\begin{equation}
\label{eq:depolarising_p(b|a)}
p_Q(b| \vec a) = \Tr\left[ \ket{b}_{T}\bra{b} \cdot \mathcal{E}_{Q}(\ket{\vec a}_T\bra{\vec a})\right] = {}_{T}\bra{b} \Big[ \mathcal{E}_{Q}(\ket{\vec a}_T\bra{\vec a}) \Big] \ket{b}_T 
,\end{equation}
where $b \equiv (b_1,b_2, \dots b_n)$ is an $n$-bit string, i.e., $b_i \in \{ 0,1 \}$, with $i\in\{1, 2, \dots n \}$. Evaluating~\eqref{eq:depolarising_p(b|a)}, we obtain
\begin{equation}
\label{eq:depolarising_p(b|a)_us}
    p_Q(b| \vec a) = \left\{\begin{array}{ll}
    \ds 1 - Q\frac{d-1}{d} & \text{ if  } b = \vec a\\ \ \\
    \ds \frac{Q}{d} & \text{ otherwise} \, .
    \end{array}\right.
\end{equation}

The action of a depolarising channel is achieved by entangling the system with the proper environment through a joint unitary evolution and then tracing the environment out. In our case, the unitary $U_{TE}^{(Q)}$ entangles the transferring qubits $T$ with the environment $E$, which consists of three parts, $E = E_1E_2E_3 = E_{123}$, such that $\mathcal H_T \cong \mathcal{H}_{E_1} \cong \mathcal{H}_{E_2}$, and $\mathcal{H}_{E_3} \cong \mathbb{C}^2$. The initial environment state is
 \begin{equation}
 \label{eq:omega_1}
     \ket{\Omega}_{E} =\ket{\Omega}_{E_{123}} = \left( \frac{1}{2^{n/2}} \sum_{i=0}^{2^n-1} \ket{i}_{E_1}\ket{i}_{E_2} \right) \otimes \left( \sqrt{1-Q} \ket{0}_{E_3} + \sqrt{Q} \ket{1}_{E_3} \right),  \end{equation}
while the controlled overall unitary is
\begin{equation}
\label{eq:controlled_depolarisation}
    U_{TE}^{(Q)} =U_{TE_{123}}^{(Q)} = \left[ I_{TE_1} \otimes \ket{0}_{E_3}\! \bra{0} + \left( \sum_{i,j=0}^{2^n-1} \ket{j}_T\ket{i}_{E_1}\!\!\bra{i}_T\bra{j}_{E_1} \right) \otimes \ket{1}_{E_3}\!\bra{1} \right] \otimes I_{E_2} .
\end{equation}
In the above controlled operation, the controlled qubit $E_3$ keeps the system $T$'s state intact (as well as that of $E_1$) if its state is $\ket{0}_{E_3}$, and swaps the states of $T$ and $E_1$ if its state is $\ket{1}_{E_3}$. For $E_1$ to be in a maximally mixed state, system $E_2$ is introduced, such that the initial state of the $E_1E_2$ is maximally entangled.

Upon receiving the transferring qubits in state
\begin{equation}\label{eq:receive_state}
\ket{\psi_1}_{ATE} = (I_A \otimes U_{TE}^{(Q)})\ket{g^{(n+1)}(0, \vec 0)}_{AT}\ket{\Omega}_E,	
\end{equation}
Bobs either reflect them, or measure them in the $Z$ basis and send them back. We analyse each case separately. 

\paragraph{Case of reflection:}
On their way back to Alice, the transferring qubits undergo another depolarisation $\mathcal{E}_{\widetilde{Q}}$ characterized by a possibly different coefficient $\widetilde{Q}$. Thus, for the new environment $\widetilde E$ 
 we have $\mathcal{H}_{\widetilde E} \cong \mathcal{H}_E$. Similarly, the joint unitary and the initial states satisfy $U_{TE}^{(\widetilde{Q})} \cong U_{TE}^{(Q)}$ and $\ket{\widetilde\Omega}_{\widetilde E} \cong \ket{\Omega}_E$, respectively. Thus, we have
\begin{equation}
\label{eq:depolarising_R_1}
\mathcal{E}_{\widetilde{Q}}(\ket{b}_T\bra{b}) \equiv \Tr_E\left[U_{TE}^{(\widetilde{Q})} \ket{b \, \widetilde\Omega}_{TE}\bra{b \, \widetilde\Omega} U_{TE}^{(\widetilde{Q})\dagger}\right] = (1-\widetilde{Q})\ket{b}_{T}\bra{b} + \frac{\widetilde{Q}}{d}I_T , 
\end{equation}
with
\begin{equation}
\label{eq:depolarising_p'(b|a)}
p_{\widetilde{Q}}(b'|b) = \Tr\left[ \ket{b'}_{T}\bra{b'} \cdot \mathcal{E}_{Q}(\ket{b}_T\bra{b})\right] = {}_{T}\bra{b'} \Big[ \mathcal{E}_{Q}(\ket{b}_T\bra{b}) \Big] \ket{b'}_T 
.\end{equation}
Evaluating~\eqref{eq:depolarising_p'(b|a)}, we obtain
\begin{equation}
\label{eq:depolarising_p'(b|a)us}
    p_{\widetilde{Q}}(b'|b) = \left\{\begin{array}{ll}
    \ds 1 - \widetilde{Q}\frac{d-1}{d} & \text{ if } b' = b\\ \ \\
    \ds \frac{\widetilde{Q}}{d} & \text{ otherwise} \, .
    \end{array}\right.
\end{equation}

Thus, in case of reflection ($R$), the evolution of the initial GHZ state  $\ket{g^{(n+1)}(0, \vec 0)}_{AT}$ before Alice's measurement is given by

\begin{eqnarray}
\label{eq:depolarising_GHZ}
     \ket{\psi^R} \!\!\! &=&\!\!\!\mathcal{E}_{R}\left(  \ket{g^{(n+1)}(0, \vec 0)}_{AT}\bra{g^{(n+1)}(0, \vec 0)}\right) \nonumber \\
     \!\!\! &=& \!\!\! \Tr_E \biggl\{\! \left[I_A \otimes \left(U_{TE}^{(\widetilde{Q})} U_{TE}^{(Q)} \right) \right] \ket{g^{(n+1)}(0, \vec 0)}_{AT}\ket{\Omega}_E\bra{\Omega}\bra{g^{(n+1)}(0, \vec 0)}_{AT} \left[I_A \otimes \left(U_{TE}^{(\widetilde{Q})}U_{TE}^{(Q)} \right) \right]^\dagger\!\biggl\},
\end{eqnarray}
while the probability of obtaining the $GHZ$ state is
\begin{eqnarray}
\label{eq:depolarising_p_GHZ}
    p_{GHZ} &=& {}_{AT}\bra{g^{(n+1)}(0, \vec 0)} \Big[ \mathcal{E}_{R}\left(\ket{g^{(n+1)}(0, \vec 0)}_{AT}\bra{g^{(n+1)}(0, \vec 0)}\right) \Big] \ket{g^{(n+1)}(0, \vec 0)}_{AT} \\
    &=&  1 - Q_{GHZ} \left(1 - \frac{1}{2^{n+1}}\right),
\end{eqnarray}
where $Q_{GHZ} \equiv Q + \widetilde Q - Q\widetilde Q$. For noiseless channels ($Q=\widetilde Q = 0$) we have $p_{GHZ}=1$. The last equality is straightforward to compute. 

\paragraph{Case of measurment:}
In case Bobs are measuring ($M$) in the $Z$ basis ($\Theta = 1$), their action is given by
\begin{equation}
    \label{eq:depolarising_Bob_Z-measurement_us}
    \left[ \sum_{\vec b = 0}^{2^n - 1} \ket{b}_T\bra{b} \otimes (\ket{b}_B\bra{\vec 0} + h.c. ) \right] \otimes I_{AE} \ket{\psi_1}_{ATE}\ket{\vec 0}_B \longmapsto \ket{{\bar\psi}^{M}}_{ATEB} ,
\end{equation}
where $\ket{\psi_1}_{ATE}$ is the state of Alice's, transferring and Eve/environment registers upon the transferring qubits arrived to Bobs, given by Equation~\eqref{eq:receive_state}. Upon their return, the transferring qubits are decohered by $\mathcal{E}_{\widetilde{Q}}$, i.e., $U_{TE}^{(\widetilde{Q})}$ entangles them with $\widetilde E$, 
\begin{equation}
    \label{eq:depolarising_Bob_Z-measurement_final}
     \left(I_A\otimes I_E\otimes I_B\otimes U_{T\widetilde E}^{(\widetilde{Q})}\right)\ket{\bar{\psi}^M}_{ATEB}\ket{\vec 0}_B\ket{\widetilde\Omega}_{\widetilde E} \longmapsto \ket{\psi^{M}}_{ATE\widetilde E B}.
\end{equation}

Finally, Alice measures her and the transferring qubits in the $Z$ basis -- she projects $\ket{\psi^{M}}_{ATE\widetilde E B}$ onto $\ket{ab}_{AT}\bra{ab}\otimes I_{EB}$. We thus get the following probabilities
\begin{equation}
    \label{eq:depolarising_Alice-Bob_Z-measurement}
    p(a, b) = \frac{\widetilde{Q}}{2^{n+1}} + \delta_{\vec a, b}\frac{1 - \widetilde{Q}}{2}.
\end{equation}
We see that 
\begin{equation}
\sum_{a = 0}^{1}\sum_{b = 0}^{2^n - 1} p(a,b) = 2^{n+1}\frac{\widetilde{Q}}{2^{n+1}} + 2\frac{1 - \widetilde{Q}}{2} = 1, 
\end{equation}
and that in the noiseless case ($Q = \widetilde Q = 0 = Q_{GHZ}$) we have $p(a, b) = \delta_{\vec a, b}/2$, i.e., $p(a, b) = 1/2$ if $b \in \{ \vec 0, \vec 1 \}$, otherwise $p(a,b) = 0$.

To obtain the secure key rate, in addition to the above probabilities~\eqref{eq:depolarising_p_GHZ} and~\eqref{eq:depolarising_Alice-Bob_Z-measurement}, we need the overlaps between the {\em appropriate} Eve/environment states. In the case of depolarising channels, those states can be straightforwardly obtained from the ``original'' Eve's {\em orthonormal} states $\{\ket{i}_{E_1}\}$, $\{\ket{j}_{E_2}\}$ and $\{\ket{0}_{E_3}, \ket{1}_{E_3}\}$ (and analogously for $\mathcal{H}_{\widetilde E}$ environment states), which significantly simplifies the calculation.

To do so, let us re-write the final state $\ket{\psi^M}_{A B T E \tilde{E}}$ from Equation~\eqref{eq:depolarising_Bob_Z-measurement_final} using unnormalized (and possibly non-orthogonal) states $\ket{{E}_{abc}}_{E\tilde{E}}$,
\begin{equation}
    \ket{\psi^M}_{A B T E \tilde{E}} = 
    \sum_{a = 0}^1 \sum_{b,c = 0}^{2^n -1} \ket{a}_A \ket {b}_{B} \ket{c}_T \ket{{E}_{abc}}_{E\tilde{E}}.
\end{equation}
There are four types of states: 
\begin{itemize}
    \item[$\bullet$] $\vec a=b=c$ (there are 2 such states):
    \begin{eqnarray*}
        2^{\frac{2n+1}{2}} \ket{{E}_{aaa}}_{E\tilde{E}} &= &\ds \sqrt{(1-Q)(1-\widetilde Q)} \left(\sum_{i,k = 0}^{2^n - 1} \ket{ii}_{E_{12}}\ket{kk}_{\widetilde{E}_{12}}\right) \ket{00}_{E_3\widetilde{E}_3} \\
        &+&\ds \sqrt{Q(1-\widetilde Q)} \ket{\vec a \vec a \, }_{E_{12}} \left(\sum_{k = 0}^{2^n - 1} \ket{kk}_{\widetilde{E}_{12}}\right) \ket{10}_{E_3\widetilde{E}_3} \\
        &+& \sqrt{(1-Q)\widetilde Q} \left(\sum_{i = 0}^{2^n - 1} \ket{ii}_{E_{12}}\right) \ket{\vec a \vec a \, }_{\widetilde{E}_{12}} \ket{01}_{E_3\widetilde{E}_3} + \sqrt{Q\widetilde Q} \ket{\vec a \vec a \, }_{E_{12}} \ket{\vec a \vec a \, }_{\widetilde{E}_{12}} \ket{11}_{E_3\widetilde{E}_3}.
    \end{eqnarray*}
 \item[$\bullet$] $\vec a=b\neq c$ (there are $2(2^n - 1)$ such states):
    \begin{eqnarray*}
        2^{\frac{2n+1}{2}} \ket{{E}_{aac}}_{E\tilde{E}} = \sqrt{(1-Q)\widetilde Q} \left(\sum_{i = 0}^{2^n - 1} \ket{ii}_{E_{12}}\right) \ket{\vec a c}_{\widetilde{E}_{12}} \ket{01}_{E_3\widetilde{E}_3} + \sqrt{Q\widetilde Q} \ket{\vec a \vec a \, }_{E_{12}} \ket{\vec a c}_{\widetilde{E}_{12}} \ket{11}_{E_3\widetilde{E}_3}.
    \end{eqnarray*}
\item[$\bullet$] $\vec a \neq b = c$ (there are $2(2^n - 1)$ such states):
    \begin{eqnarray*}
        2^{\frac{2n+1}{2}} \ket{{E}_{abb}}_{E\tilde{E}} = \sqrt{Q(1-\widetilde Q)} \ket{\vec a b}_{\widetilde{E}_{12}}\left(\sum_{k = 0}^{2^n - 1} \ket{kk}_{E_{12}}\right)  \ket{01}_{E_3\widetilde{E}_3} + \sqrt{Q\widetilde Q} \ket{\vec a b}_{E_{12}} \ket{bb}_{\widetilde{E}_{12}} \ket{11}_{E_3\widetilde{E}_3}.
    \end{eqnarray*}
\item[$\bullet$] $\vec a \neq b \neq c$ (there are $2([2^n - 1)2^n - (2^n - 1)] = 2[2^{2n} - (2^n - 1)]$ such states; note that  it is possible~$a = c$):
    \begin{eqnarray*}
        2^{\frac{2n+1}{2}} \ket{{E}_{abc}}_{E\tilde{E}} = \sqrt{Q\widetilde Q} \ket{\vec ab}_{E_{12}} \ket{bc}_{\widetilde{E}_{12}} \ket{11}_{E_3\widetilde{E}_3}.
    \end{eqnarray*}    
\end{itemize}

\noindent The (squared) norms of the above states are:
\begin{itemize}
    \item[$\bullet$] $\displaystyle _{E\tilde{E}}\braket{{E}_{aaa}|{E}_{aaa}}_{E\tilde{E}} = \frac{(1-Q)(1-\widetilde{Q})}{2} + \frac{Q(1-\widetilde{Q}) + (1-Q)\widetilde{Q}}{2^{n+1}} + \frac{Q\widetilde{Q}}{2^{2n+1}}.$ 
    \item[$\bullet$] $\displaystyle _{E\tilde{E}}\braket{{E}_{aac}|{E}_{aac}}_{E\tilde{E}} = \frac{(1-Q)\widetilde{Q}}{2^{n+1}} + \frac{Q\widetilde{Q}}{2^{2n+1}}.$ 
    \item[$\bullet$] $\displaystyle _{E\tilde{E}}\braket{{E}_{abb}|{E}_{abb}}_{E\tilde{E}} = \frac{Q(1-\widetilde{Q})}{2^{n+1}} + \frac{Q\widetilde{Q}}{2^{2n+1}}.$ 
    \item[$\bullet$] $\displaystyle _{E\tilde{E}}\braket{{E}_{abc}|{E}_{abc}}_{E\tilde{E}} = \frac{Q\widetilde{Q}}{2^{2n+1}}.$
\end{itemize}

Finally, all states are orthogonal between each other, apart from $\ket{{E}_{000}}_{E\tilde{E}}$ and $\ket{{E}_{111}}_{E\tilde{E}}$, whose overlap is:
\begin{equation*}
    _{E\tilde{E}}\braket{{E}_{000}|{E}_{111}}_{E\tilde{E}} = \frac{(1-Q)(1-\widetilde{Q})}{2}.
\end{equation*}

Considering the fact that all vectors, save $_{E\tilde{E}}\braket{E_{000}|E_{111}}_{E\tilde{E}}$ are orthogonal, it is easy to simplify Equations~\eqref{eq:new-entropy-bound} and~\eqref{eq:new-lambda} to the following:
\begin{equation}
    S(A|E) \ge \left(\frac{(1-Q)(1-\widetilde{Q})}{2} + \frac{Q(1-\widetilde{Q}) + (1-Q)\widetilde{Q}}{2^{n+1}} + \frac{Q\widetilde{Q}}{2^{2n+1}}\right)\big(1 - H(\lambda(0\,0\, 1\,1)\big), \texttt{ and}
\end{equation}
\begin{equation}
    \lambda(0\,0\,1\,1) = \frac{1}{2}\left(1 + \frac{(1-Q)(1-\widetilde{Q})}{2\ds \left(\frac{(1-Q)(1-\widetilde{Q})}{2} + \frac{Q(1-\widetilde{Q}) + (1-Q)\widetilde{Q}}{2^{n+1}} + \frac{Q\widetilde{Q}}{2^{2n+1}}\right)}\right).
\end{equation}

The {\em error correction leakage } will be $\max_jH(A|B_j)$ (see Equation \eqref{eq:key-rate}) which can be readily seen to be equal to $h(Q_{\texttt{Bob}})$, where $h(x)$ is the binary Shannon entropy function and where $Q_{\texttt{Bob}}$ is the noise observed by one of the Bobs (since the noise in all channels is assumed to be the same, it does not matter, in this case, which Bob we consider).  This value can be seen to be:
\begin{equation}
    Q_{\texttt{Bob}} = \frac{1}{2}\sum_{b = 0}^{2^n -1}(p(0b|1) + p(1b|0)) = 2^{n-1}Q/d = Q/2.
\end{equation}
which is enough to compute the bound of the key rate expressed in Equation~\eqref{eq:key-rate}.
Thus the final key-rate equation is:
\begin{equation}
r \geq S(A|E) - H\left(Q_{\texttt{Bob}}\right) \equiv r_{\texttt{min}}.
\end{equation}

\begin{figure}[t]
    \centering
    \includegraphics[width=.9\linewidth]{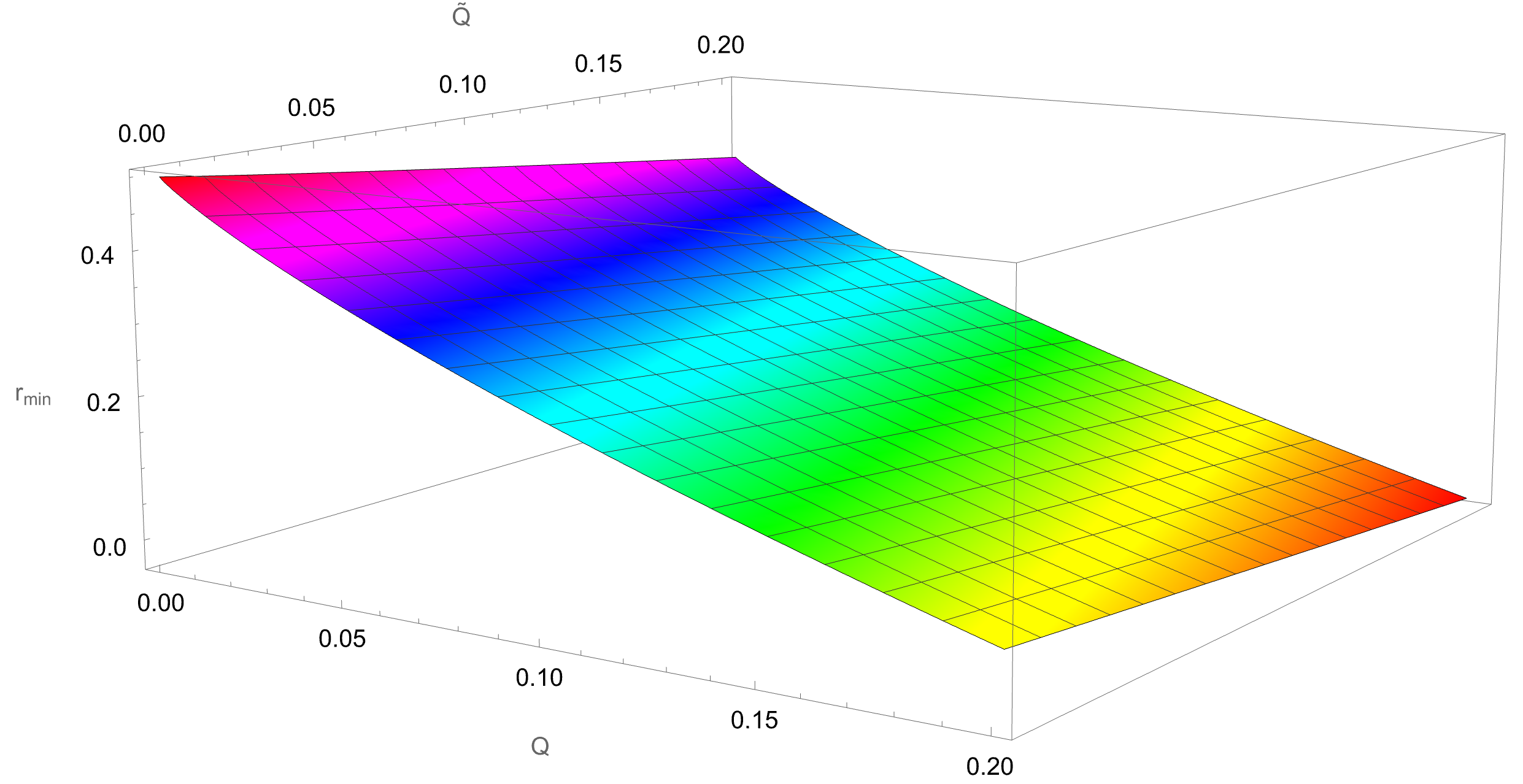}
    \caption{Key rate lower bounds for depolarization channel with $n=10$ Bobs.
    }
    \label{fig:keyrate_3D}
\end{figure}
In Figures~\ref{fig:keyrate_3D},~\ref{fig:keyrate_Q=tildeQ},~\ref{fig:keyrate_Q=0} and~\ref{fig:keyrate_tildeQ=0}, we present the plots for the key rate for the depolarizing channel. In Figure~\ref{fig:keyrate_3D} we present $r_{\texttt{min}}$ as a function of $Q$ and $\widetilde{Q}$ for the case of $n=10$ Bobs. In Figure~\ref{fig:keyrate_Q=tildeQ} we present $r_{\texttt{min}}$ in the special case of $Q = \widetilde{Q}$, for $n = 3,5,7$ Bobs. We see that the range of values of $Q$ grows, albeit slowly, with the increase of the number of Bobs, approaching the value of $Q = \widetilde{Q} = 0.2$.
This increase might be because a higher number of Bobs leads to more interactions between the adversary and the agents, making it easier to detect eavesdropping. Furthermore, since these values represent only the lower bounds of the key rate, it is possible that these bounds become tighter as the number of Bobs increases. Nevertheless, the provided bounds are sufficient as a proof-of-concept and can also be used in an experimental setup.

 Finally, in order to further analyse the asymmetry of Figure~\ref{fig:keyrate_3D}, on Figures~\ref{fig:keyrate_Q=0} and~\ref{fig:keyrate_tildeQ=0} we present the values of $r_{\texttt{min}}$ when $Q$ and $\widetilde{Q}$ are set to zero, respectively, for $n = 3,5,7$ Bobs. We see that when the forward channel is turned off ($Q=0$), we have positive key rate for any value of $\widetilde{Q}$. In contrast to this, when the reflecting qubits are intact ($\widetilde{Q}=0$), one obtains positive lower bounds for the key rate only for $Q < 0.25$. This clearly shows that the forward channel has higher impact on the protocol performance, which can be explained by the fact that the key is generated after the transferring qubits sent by Alice through the forward channel reach Bobs. Thus, only the forward channel affects the key generation, while the noise exhibited by the qubits on their way back to Alice only affects the ``visibility'' of the final key. But as long as the forward channel is noiseless and the Bobs share the perfect key among themselves, possible noise on the qubits' way back to Alice can only diminish the overall key rate (unlike Bobs, Alice might not share the same raw key), but the agents can always extract a non-trivial secure key upon error correction and privacy amplification.

\begin{figure}[t]
    \centering
    \includegraphics[width=.5\linewidth]{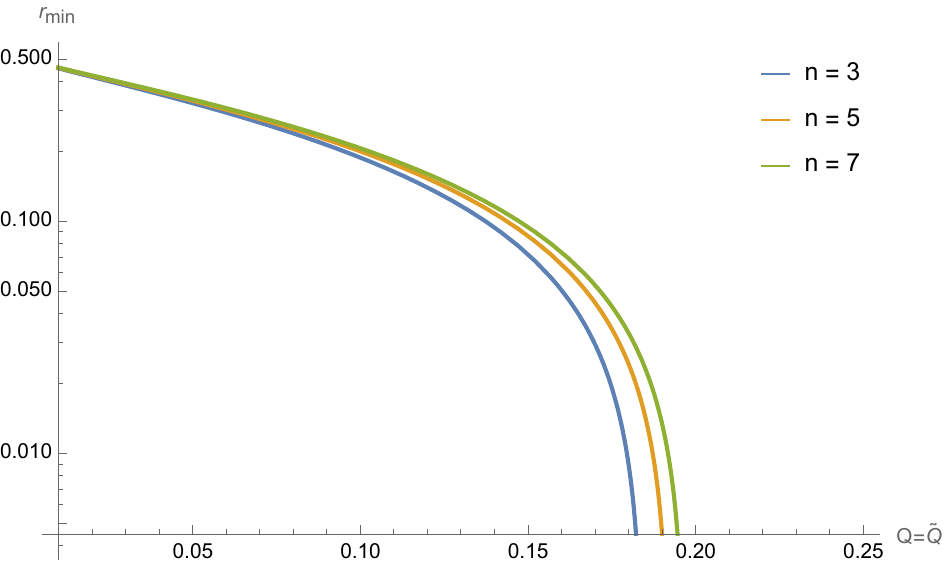}
    \caption{Key rate lower bounds for depolarization channel for several number $n$ of Bobs with $Q = \widetilde Q$.}
    \label{fig:keyrate_Q=tildeQ}
\end{figure}

\begin{figure}[t]
\begin{subfigure}{.5\textwidth}
    \centering
    \includegraphics[width=1\linewidth]{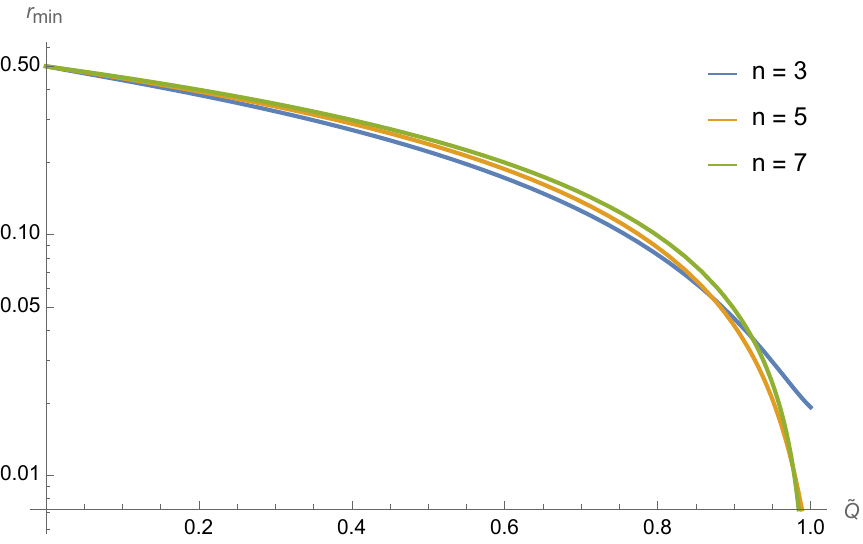}
    \caption{}
    \label{fig:keyrate_Q=0}
\end{subfigure}
\begin{subfigure}{.5\textwidth}
    \centering
    \includegraphics[width=1\linewidth]{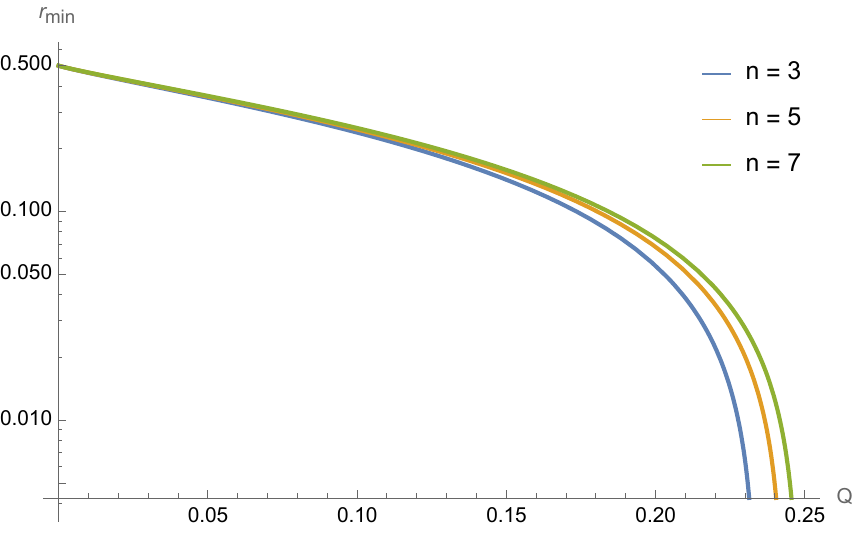}
    \caption{}
    \label{fig:keyrate_tildeQ=0}
\end{subfigure}
\caption{Key rate lower bounds for depolarization channel for several number $n$ of Bobs  with (a) $Q=0$, and (b) $\widetilde Q = 0$.}
\end{figure}

\section{Conclusions}
\label{sec:conclusions}

In this work, we have proposed a semi-quantum conference key agreement (SQCKA) protocol utilizing GHZ states. 
Unlike previous ``fair key agreement'' protocols, where the goal is to ensure all parties contribute equally towards the randomness of the final raw key, our protocol derives randomness directly from measuring a quantum state prepared by one party (or an adversary) as is typically done in most standard (S)QKD protocols.  This leads to a simpler protocol for parties to execute.  Beyond fair key protocols, our work differs from other semi-quantum group key protocols in that we do not require a cycle topology as in~\cite{xian2009quantum,ye2023circular} (which creates additional opportunities for an attacker); and for those protocols not reliant on a cycle topology~\cite{pan2024multi,zhou2019multi}, we provide the first information theoretic analysis of security, computing the key rate in terms of noise in the channel, for any number of parties.  Prior work in non-cycle topologies investigated security in terms of \emph{robustness} (where it is shown that any attack which allows Eve to learn something about the raw key measurements, such an attack can be detected with non-zero probability) or performed an information-theoretic security proof, but only for two or three party variants of the protocol.  Our proof works for any number of parties.

Our contribution introduces a novel SQCKA protocol with an information-theoretic security proof against collective attacks, all while eliminating the need for a trusted mediator. This advancement simplifies the practical implementation of SQCKA protocols, making them more accessible and feasible for real-world applications. The novelty of our work lies in the innovative security analysis and the assurance of robust security guarantees without the need for complex network topologies or third-party trust.

Looking ahead, there are several directions for future research. Developing a general stand-alone security proof that covers a wider range of attack models beyond collective attacks would significantly enhance the robustness of our protocol. Investigating the possibility of a device-independent security proof could ensure security even when the devices used are untrusted or imperfect. Formulating a composable security proof would guarantee that the security of our protocol holds when integrated with other cryptographic protocols. Finally, exploring practical implementation strategies and conducting experimental validations would be essential to assess the real-world performance and robustness of the proposed SQCKA protocol. Addressing these future research directions will strengthen the theoretical foundations of SQCKA protocols and pave the way for their practical deployment in secure communication networks.

\section*{Acknowledgements} 

\noindent RB, PM, and NP's work is funded by FCT/MECI through national funds and, when applicable, co-funded EU funds under UID/50008: Instituto de Telecomunica\c{c}\~{o}es.  NP acknowledges the FCT Est\'imulo ao Emprego Cient\'ifico grant no. CEECIND/04594/2017/CP1393/CT000.   AS acknowledges FCT, LASIGE Research Unit, ref. UID/000408/2025.  RB, PM, NP and AS also acknowledge FCT projects QuantumPrime reference PTDC/EEI-TEL/8017/2020.  WOK would like to acknowledge support from the NSF under grant number 2143644.

\bibliographystyle{unsrt}
\bibliography{bibtex}

\end{document}